\documentclass[pra,twocolumn,superscriptaddress,nofootinbib]{revtex4}

\usepackage{amsfonts,amssymb,amsmath}
\usepackage{graphics,graphicx,epsfig}
\usepackage{amsthm}

\def\identity{\leavevmode\hbox{\small1\kern-3.8pt\normalsize1}}

\newtheorem{lemma}{Lemma}

\newcommand{\ket}[1]{| #1 \rangle}

\renewcommand{\epsilon}{\varepsilon}

\usepackage{color}

\begin{document}

\title{Small sets of complementary observables}
\date{\today}

\author{M. Grassl}
\affiliation{Max Planck Institute for the Science of Light, 91058 Erlangen, Germany}

\author{D. McNulty}
\affiliation{Department of Applied Mathematics, Hanyang University (ERICA), 55 Hanyangdaehak-ro, Ansan, Gyeonggi-do, 426-791, Korea}

\author{L. Mi\v{s}ta Jr}
\affiliation{Department of Optics, Palack\'y University, 17 listopadu 12, 771 46 Olomouc, Czech Republic}

\author{T. Paterek}
\affiliation{School of Physical and Mathematical Sciences, Nanyang Technological University, 637371 Singapore, Singapore}
\affiliation{Centre for Quantum Technologies, National University of Singapore, 117543 Singapore, Singapore}
\affiliation{MajuLab, CNRS-UNS-NUS-NTU International Joint Research Unit, UMI 3654 Singapore, Singapore}

\begin{abstract}
Two observables are called complementary if preparing a physical
object in an eigenstate of one of them yields a completely random
result in a measurement of the other.  We investigate small sets of
complementary observables that cannot be extended by yet another
complementary observable.  We construct explicit examples of the
unextendible sets up to dimension $16$ and conjecture certain small
sets to be unextendible in higher dimensions.  Our constructions
provide three complementary measurements, only one observable away
from the ultimate minimum of two observables in the set.  Almost all
of our examples in finite dimension allow to discriminate pure states
from some mixed states, and shed light on the complex topology of the Bloch space of higher-dimensional quantum systems.
\end{abstract}

\maketitle

Quantum mechanics restricts precision with which certain
(incompatible) observables can be measured, as characterised, e.g., by
the Heisenberg uncertainty relation.  Here we study complementary
observables which are, per definition, the most incompatible ones.
The idea is clearly illustrated if one considers sequential
measurements of such observables.  After measuring one of them and
subjecting the resulting object to a measurement of any other
complementary observable the results are completely random and
independent of the particular outcome observed in the first
measurement.  The eigenstates of such observables form so-called
mutually unbiased bases (MUBs), and due to their elementary role in
quantum physics they have received considerable
attention~\cite{IJQI.12.535}.  Typically, researchers ask about the
maximal number of complementary observables. It is known that the
maximal possible number of $d+1$ of them for a $d$-level object can be
reached when $d$ is a power of a prime.  Compelling evidence suggests
that the maximum is strictly smaller for other dimensions.

Here we focus on the opposite question: what is the minimal number of
complementary measurements?  It has been observed in various
construction methods for MUBs that if one starts ``wrongly'', it is
impossible to construct all $d+1$ complementary measurements; one gets
stuck at some strictly smaller
number~\cite{0406175,0502024,PhysRevA.79.012109,PhysRevA.79.052316,JPA.42.245305,JPA.45.102001,JPhysA.46.105301}.
Such sets of complementary observables are called unextendible and
they are known to exist in small dimensions as well as in all
dimensions $p^2$, where $p$ is a prime satisfying $p \equiv 3 \mod
4$~\cite{1502.05245}, and in all dimensions $2^m$ for even
$m$~\cite{1604.04797}.  Other examples are known if the notion of
unextendibility is somewhat
relaxed~\cite{QuantInfComp.14.823,1407.2778,1508.05892}.  From this
perspective we are interested in small sets of unextendible MUBs.

Our results for small dimensions are summarised in
Table~\ref{TAB_SUMMARY}.  The sets that we construct contain no more
than three measurements.  While we were not able to prove that they
are the minimal sets, they contain only one more measurement from the
ultimate minimum\footnote{There always exists at least one MUB with
  respect to a given basis~\cite{PNAS.46.560}.}.  Quite surprisingly
this is also the case for certain prime dimensions where we provide
explicit examples in small dimensions, and a candidate set for more
general dimensions.  We also discuss various approaches to
unextendibility which reveal the existence of unextendible sets of
complementary observables in infinitely many dimensions.

\section{Problem and summary}

\begin{figure}[!b]
\includegraphics[width=0.48\textwidth]{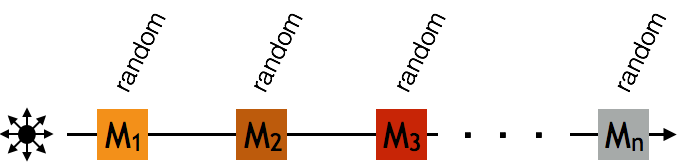}
\caption{The problem studied here can be visualised as follows.  A
  physical system prepared initially in a completely mixed state is
  interrogated by a sequence of $n$ players.  Each of them chooses to
  perform the corresponding projective measurement $M_j$ or forwards
  the system untouched to the next player.  It is required that all
  the measurements have completely random outcomes, independent of the
  results of the previous measurement and which of the measurements
  have been performed.  We ask about the minimal number of
  measurements (respectively players) in this sequence for a $d$-level
  input state.  See Table~\ref{TAB_SUMMARY} for a summary in small
  dimensions and the main text for other dimensions.}
\label{FIG_SCENARIO}
\end{figure}

Consider the scenario depicted in Fig.~\ref{FIG_SCENARIO}.  We are
given a sequence of $n$ players. Starting with a
completely mixed state, player $j$ either performs the measurement
$M_j$ or passes the state unchanged to the next player.  It is
required that all the measurements have completely random outcomes,
independent of the results of the previous measurement and which of
the measurements have been performed.  This is formally encoded by the
property of mutual unbiasedness:
\begin{equation}
|\langle k_j | k'_{j'} \rangle|^2 = 1/d \quad \textrm{ for } \quad  j \ne j',
\end{equation}
where $\ket{k_j}$ is the $k$th state of the eigenbasis of the $j$th
observable.  The question we investigate is how small can $n$ be such
that no player can be added.

\begin{table}[!t]
\begin{tabular}{c | c}
\hline \hline
dimension & cardinality of minimal set \\ \hline \hline
$2$ & $3$ \\ \hline
$3$ & $4$ \\ \hline
$4$ & $3$ \\ \hline
$5$ & $6$ \\ \hline
$6$ & $2$ \\ \hline
$7$ & \textcolor{blue}{$\le$ 3} \\ \hline
\vdots & \raisebox{0.6ex}{\textcolor{blue}{$\le 3$}} \\[0.6ex] \hline
$16$ & \textcolor{blue}{$\le$ 3} \\ \hline
$\infty$ & \textcolor{blue}{$3$ for observables $r( \cos\theta \, \hat q + \sin\theta \, \hat p)$} \\
\hline \hline
\end{tabular}
\caption{Summary of the results.  Note the intriguing case of
  dimension $6$ for which there exists a set with only two
  measurements.  All the numbers in blue are derived here.}
\label{TAB_SUMMARY}
\end{table}

The results are presented in Table~\ref{TAB_SUMMARY}, showing that the
number $n$ of measurements is very small.  Apart from fundamental
interest this also has practical applications in witnessing various
physical properties.  We briefly describe how the unextendible sets
help in witnessing the degree of purity of a state.

A set of MUBs is called \emph{strongly unextendible} if there is not
even a single vector unbiased to all the vectors in the set.  In this
case there is not a single pure quantum state that gives rise to
completely random results when measured with the observables from the
unextendible set.  In other words, the entropy of the measurement
results summed over the observables has an upper bound
$H_\textrm{p}<n\log_2 d$.  Hence, observing larger entropy than
$H_\textrm{p}$ indicates that a mixed state is being measured.  All
our examples in Table~\ref{TAB_SUMMARY}, except for $d=6$, form
strongly unextendible MUBs. Note that any \emph{pair} of MUBs can
never be strongly unextendible since for any two bases of a
$d$-dimensional Hilbert space there exist at least $2^{d-1}$ vectors
unbiased to both \cite{korzekwa2014operational}.

We would also like to note that the existence of unextendible MUBs
provides information on the complex topology of the ``Bloch'' space of
quantum states of higher-dimensional quantum systems. A density matrix
describing a pure state of a $d$-level quantum system can be written
in the operator basis of generalised Pauli matrices
\cite{Jakobczyk_01,Kimura_03}. In this ``Bloch'' representation a pure
quantum state is mapped to a normalised vector with $d^2-1$ real
coefficients. It turns out that the resulting set of states is rather
complex for dimensions $d \ge 3$, and unextendible MUBs provide a hint
for why this is the case.  While mutually unbiased states are mapped
to orthogonal Bloch vectors, orthogonal quantum states are mapped to
non-orthogonal Bloch vectors with an angle $\cos\theta =
-\frac{1}{d-1}$ between them. A basis in the Hilbert space translates
to a regular ($d-1$)-dimensional simplex in the Bloch space
\cite{Appleby_09}. Thus the nonexistence of a vector that is unbiased
to a set of MUBs means that there is not a single unit Bloch vector of
a physical pure quantum state in the orthogonal complement to the
space spanned by the Bloch vectors of unextendible MUBs. These
orthogonal complements are of high dimensionality given by $d^2 - 1 -
m(d-1)$, where $m$ is the number of bases in the unextendible set. For
example, in the Bloch space of a four-level system, there exists a
$6$-dimensional subspace (out of $15$ dimensions in total) with no unit Bloch vector representing a
physical pure state.

The paper is organised around Table~\ref{TAB_SUMMARY}.  In
Section~\ref{SEC_SMALL} we gather previous results from which one can
infer the size of the minimal set of complementary observables up to
dimension $6$ (including $6$).  The new results in the Table are
obtained by different methods which are described in subsequent
sections.  We first arrive at various candidate unextendible sets and
next show that they are indeed unextendible by computer-aided
algebraic methods.  Finally, we discuss various approaches to
unextendibility and use them to prove the existence of unextendible
sets in many new dimensions.\footnote{Note that the word
  ``unextendible'' in connection with bases in Hilbert spaces is
  sometimes used in the literature in a different context.  An
  unextendible product basis, as in Ref.~\cite{PhysRevLett.82.5385},
  denotes an incomplete orthogonal product basis whose complementary
  subspace contains no product state.  Similarly, an unextendible
  maximally entangled basis, as in Ref.~\cite{PhysRevA.84.042306},
  denotes an incomplete set of orthonormal maximally entangled states
  which have no additional maximally entangled vectors orthogonal to
  all of them.}

\section{\boldmath Small dimensions ($d \le 6$)}
\label{SEC_SMALL}

Let us begin by recalling the cardinalities of minimal sets of
complementary observables in dimensions $d \le 6$.  If the standard
basis is represented by the identity matrix, then the matrices
representing all bases unbiased to it must have all entries of the
same modulus.  The results below follow from studies of these
so-called complex Hadamard matrices. Note that two Hadamard matrices
are equivalent if one transforms into the other by multiplication with
permutation and diagonal matrices, and two sets of MUBs are equivalent
if a unitary matrix maps one set to the other.

In dimension two we have no unextendible MUBs. All $2\times 2$
Hadamard matrices are equivalent to the Fourier matrix $F_2$, and only
one triple of MUBs exists. These three bases correspond to the
eigenstates of the Pauli operators $\sigma_x$, $\sigma_y$, and
$\sigma_z$.

Similarly, in dimension three no unextendible MUBs exist.  All
Hadamard matrices are equivalent to the Fourier matrix $F_3$, and six
vectors are mutually unbiased to the pair $\{\openone, F_3 \}$.  These
six vectors can be partitioned into two bases.  Together with
$\openone$ and $F_3$ they form a unique quadruple of MUBs that
contains two inequivalent triples of MUBs.

The smallest case in which unextendible MUBs appear is dimension four,
where an infinite family of cardinality three exists.  In this
dimension all pairs of MUBs are equivalent to a member of the
one-parameter set $\{ \openone, F_4(a)\}$, where $F_4(a)$ is the
one-parameter Fourier family of $4\times 4$ matrices.  Each of these
pairs extends to a mutually unbiased triple $\{\openone, F_4(a),
H_2(b,c)\}$, producing a three-parameter set~\cite{0907.4097}.  If $a
= \pi/2$, then unique sets of four and five MUBs exist (note that a
set of $d$ MUBs in dimension $d$ is always extendible~\cite{PAMS.141.1963}), and other
values of $a$ give rise to an unextendible set.  A simple example of
an unextendible triple of MUBs is the following:
\begin{eqnarray}
B_1 = \left(
\begin{array}{cccc}
1 & 0 & 0 & 0 \\
0 & 1 & 0 & 0 \\
0 & 0 & 1 & 0 \\
0 & 0 & 0 & 1
\end{array}
\right), \qquad
B_2 = \frac{1}{2} \left(
\begin{array}{cccc}
1 & 1 & 1 & 1 \\
1 & i & -1 & -i \\
1 & - 1 & 1 & -1 \\
1 & -i & -1 & i
\end{array}
\right), \nonumber
\end{eqnarray}
\begin{eqnarray}
B_3 = \frac{1}{2} \left(
\begin{array}{cccc}
1 & e^{i a} & 1 & -e^{i a} \\
1 & -e^{ia} & 1 & e^{i a} \\
1 & e^{i b} & -1 & e^{i b} \\
1 & -e^{ib} & -1 & - e^{i b}
\end{array}
\right),
\end{eqnarray}
with $a,b \in [0, \pi)$, $a \ne \pi/2$, and where the rows of each
  matrix represent the orthogonal vectors of each basis.

In dimension five there are no unextendible MUBs.  According
to~\cite{Haagerup} there is a unique choice of Hadamard matrix, $F_5$,
and exactly $20$ vectors are mutually unbiased to the pair $\{\openone,
F_5\}$. These vectors can be partitioned into four orthonormal bases,
resulting in unique combination of six MUBs~\cite{0907.4097}.

It turns out that our problem is also solved for dimension six: there
is a pair of complementary measurements which do not extend to a
triple~\cite{PhysRevA.79.052316}.  This pair corresponds to the set
$\{\openone, S_6 \}$, where $S_6$ is the isolated Hadamard
matrix~\cite{Butson62}. %~\cite{Moorhouse01,Tao04}
\begin{equation}
S_6 = \frac{1}{\sqrt{6}}
\left(
\begin{array}{cccccc}
1 & 1 & 1 & 1 & 1 & 1 \\
1 & 1 & \omega_3 & \omega_3 & \omega_3^2 & \omega_3^2 \\
1 & \omega_3 & 1 & \omega_3^2 & \omega_3^2 & \omega_3 \\
1 & \omega_3 & \omega_3^2 & 1 & \omega_3 & \omega_3^2 \\
1 & \omega_3^2 & \omega_3^2 & \omega_3 & 1 & \omega_3 \\
1 & \omega_3^2 & \omega_3 & \omega_3^2 & \omega_3 & 1
\end{array}
\right),
\end{equation}
with $\omega_3 = \exp(2 \pi i/ 3)$. A complete classification of
Hadamard matrices in dimension six is unknown, but other minimal sets
may exist.

\section{\boldmath Dimensions $7$, $11$, and primes $p\equiv3 \mod 4$}

We begin with a lemma that holds in arbitrary dimensions.  It will
then be utilised in the construction of a particular triple of MUBs
that will be shown to be unextendible by computer-aided techniques.

The lemma utilises the Fourier basis, which is defined as
\begin{equation}
| \tilde\jmath \rangle = F \ket{j} = \frac{1}{\sqrt{d}} \sum_{m = 0}^{d-1} \omega_d^{j m} \ket{m},
\label{FOURIER}
\end{equation}
where $\omega_d = \exp(2 \pi i / d)$ is a $d$th root of unity,
$\{\ket{j}\}$ denotes the standard basis, and $F$ denotes the Fourier
matrix with elements $F_{jm} = \frac{1}{\sqrt{d}} \omega_d^{j m}$,
where the rows and columns are enumerated as $j,m = 0,\dots, d-1$.

\begin{lemma}
Let $\ket{v^{(0)}} = \sum_{m=0}^{d-1} v_m \ket{m}$ be a state that is
unbiased with respect to the Fourier basis.  Then the states
$\ket{v^{(k)}} = \sum_{m=0}^{d-1} v_m \ket{(m+k) \bmod d}$ form an
orthonormal basis that is unbiased with respect to the Fourier basis.
\label{LEM_FOURIER}
\end{lemma}
\begin{proof}
First we prove that the unbiasedness of the vector $\ket{v^{(0)}}$ implies
that all vectors $\ket{v^{(k)}}$ are unbiased to the Fourier basis.
Indeed,
\begin{eqnarray}
  \left| \langle \tilde\jmath | v^{(k)} \rangle \right|^2
  & = & \left| \frac{1}{\sqrt{d}} \sum_{m=0}^{d-1} v_m \omega_d^{- j (m+k)} \right|^2 \nonumber \\
  & = & \left| \frac{1}{\sqrt{d}} \sum_{m=0}^{d-1} v_m \omega_d^{- j m} \right|^2 = \frac{1}{d},
\label{SHIFT_UNBIAS}
\end{eqnarray}
where the step in between the lines follows from $|\omega_d^{ - j k}|
= 1$, and the last step expresses our assumption about the
unbiasedness of the vector $\ket{v^{(0)}}$.

It remains to be shown that vectors $\{ \ket{v^{(k)}} \}$ form a
basis.  To this end let us first compute their inverse Fourier
transform:
\begin{alignat}{5}
| \tilde{v}^{(k)} \rangle = F^{\dag} \ket{v^{(k)}}&{} = \frac{1}{\sqrt{d}} \sum_{l,m = 0}^{d-1} v_m \omega_d^{-l(m+k)} \ket{l}\nonumber\\
 &{}= \sum_{l =0}^{d-1} \tilde v_l \omega_d^{-k l} \ket{l},
\end{alignat}
where we introduced $\tilde v_l = \frac{1}{\sqrt{d}} \sum_{m =
  0}^{d-1} v_m \omega_d^{-m l}$.  Since the Fourier transform
preserves the inner products, it is sufficient to show that the
transformed vectors form a basis:
\begin{equation}
\langle \tilde{v}^{(k)} | \tilde{v}^{(k')} \rangle = \sum_{l = 0}^{d-1} \omega_d^{l (k - k')} |\tilde v_l|^2 = \frac{1}{d} \sum_{l=0}^{d-1} \omega_d^{l (k-k')} = \delta_{k k'},
\end{equation}
where, due to \eqref{SHIFT_UNBIAS}, we have used that $|\tilde v_l|^2
= 1/d$ for all values of $l$.
\end{proof}

Let us now restrict ourselves to prime dimensions $p$ which are equal
to $p \equiv 3 \mod 4$.

\begin{lemma}
\label{LEM_3M4}
Let $p$ be a prime with $p \equiv 3 \mod 4$.  Then the states
$\ket{v^{(k)}} = \frac{1}{\sqrt{p}} \sum_{m = 0}^{p-1} v_m \ket{(m+k)
  \bmod p}$ with
\begin{equation}
v_m = \Big\{
\begin{array}{rl}
\frac{-(p-1)+ 2 \sqrt{-p}}{p+1} & \textrm{ if } m \in \mathcal{Q}, \\
1 & \textrm{ if } m \in \mathcal{N},
\end{array}
\end{equation}
form a mutually unbiased basis with respect to both the standard and
the Fourier basis.  The sets $\mathcal{Q}$ and $\mathcal{N}$ denote
the quadratic residues (including zero) and non-residues modulo $p$,
respectively.
\end{lemma}
\begin{proof}
Clearly, for all $m$ we have $|v_m|^2 = 1$, i.e., all vectors
$\ket{v^{(k)}}$ are mutually unbiased with respect to the standard basis.
We now prove that $\ket{v^{(0)}}$ is unbiased with respect to the Fourier
basis.  Consider
\begin{equation}
\tilde v_{\tilde\jmath} = \langle \tilde\jmath | v^{(0)} \rangle = \frac{1}{p} \sum_{m=0}^{p-1} v_m \omega_p^{- j m}.
\end{equation}
For $j=0$ we find
\begin{eqnarray}
\tilde v_{\tilde 0} & = & \frac{1}{p} \sum_{m = 0}^{p-1} v_m = \frac{1}{p} \left( \frac{p-1}{2} + \frac{p+1}{2} \,\, \frac{2 \sqrt{-p}-(p-1)}{p+1} \right) \nonumber \\
& = & \frac{\sqrt{-p}}{p},
\end{eqnarray}
where we have used the fact that the number of elements in
$\mathcal{N}$ ($\mathcal{Q}$) is equal to $\frac{p-1}{2}$
($\frac{p+1}{2}$).  This shows that indeed $|\tilde v_{\tilde 0}|^2 =
1/p$.  In order to find the other coefficients, we first note that for
$j \not\equiv 0 \mod p$:
\begin{eqnarray}
\sum_{m \in \mathcal{Q}} \omega_p^{- j m } & = & \frac{1-\sqrt{-p}}{2}, \\
\text{and}\quad
\sum_{m \in \mathcal{N}} \omega_p^{- j  m } & = & \frac{-1+\sqrt{-p}}{2}.
\end{eqnarray}
Thus, for $j \not\equiv 0\mod p$ we have
\begin{eqnarray}
\tilde v_{\tilde\jmath} & = & \frac{1}{p} \left( \frac{-1+\sqrt{-p}}{2} + \frac{1-\sqrt{-p}}{2} \, \, \frac{2 \sqrt{-p} -(p-1)}{p+1} \right) \nonumber \\
& = & \frac{\sqrt{-p}}{p},
\end{eqnarray}
revealing that $\ket{v^{(0)}}$ is unbiased to the Fourier basis.
The proof concludes by utilising Lemma~\ref{LEM_FOURIER}.
\end{proof}

A concrete example of the basis introduced in Lemma~\ref{LEM_3M4} for
dimension $7$ reads as follows:
\begin{equation}
A = \frac{1}{\sqrt{7}}
\left(
\begin{array}{ccccccc}
\alpha_7 & \alpha_7 & \alpha_7 & 1 & \alpha_7 & 1 & 1 \\
1 & \alpha_7 & \alpha_7 & \alpha_7 & 1 & \alpha_7 & 1 \\
1 & 1 & \alpha_7 & \alpha_7 & \alpha_7 & 1 &\alpha_7 \\
\alpha_7 & 1 & 1 & \alpha_7 & \alpha_7 & \alpha_7 & 1 \\
1 & \alpha_7 & 1 & 1 & \alpha_7 & \alpha_7 & \alpha_7 \\
\alpha_7 & 1 &\alpha_7 & 1 & 1 & \alpha_7 & \alpha_7 \\
\alpha_7 & \alpha_7 & 1 & \alpha_7 & 1 & 1 & \alpha_7
\end{array}
\right),
\label{A7_UMUB}
\end{equation}
where $\alpha_7 \equiv \frac{\sqrt{-7} - 3}{4}$.  To see that the rows
or columns correspond to orthogonal states note that $\alpha_7 +
\alpha_7^* = -\frac{3}{2}$.

We have shown that the standard basis, the Fourier basis, and the
basis of Lemma~\ref{LEM_3M4} form a triple of MUBs.  We verified by
computer-aided techniques that this triple is unextendible in
dimensions $p = 7$ and $p = 11$.  We conjecture that for any prime $p$
with $p \equiv 3 \mod 4$, the triple obtained in this way is
unextendible.  Note that in prime dimensions smaller than $7$ there
are no unextendible sets of MUBs and, furthermore, the only existing
sets of MUBs are maximal.  From this perspective the existence of
basis $A$ is unexpected.

\section{\boldmath Dimension $13$ and primes $p\equiv 1 \mod 4$}

One can proceed along similar lines for other prime dimensions.
This time the candidate basis is given by the following lemma.

\begin{lemma}
\label{LEM_GABIDULIN}
(See also Ref.~\cite{ProbInfTrans.38.255})
Let $p$ be a prime with $p \equiv 1 \mod 4$. 
Define
\begin{equation}
z_0 =cos\theta + i \sin\theta \quad \textrm{ with } \quad \cos\theta = \frac{\sqrt{p} - 1}{p - 1}.
\end{equation}
Then the states $\ket{v^{(k)}} = \frac{1}{\sqrt{d}} \sum_{m = 0}^{p-1} v_m
\ket{(m+k)\bmod p}$, where
\begin{equation}
  v_m =\begin{cases}
  1,     & \text{if $m=0$;} \\
  z_0,   & \text{if $m \in \mathcal{Q}\setminus\{0\}$;} \\
  1/z_0, & \text{if $m \in \mathcal{N}$,}
  \end{cases}
\end{equation}
form a mutually unbiased basis with respect to both the standard and
the Fourier basis.
\end{lemma}
\begin{proof}
Since $z_0$ lies on the unit circle the vectors $\ket{v^{(k)}}$ are
clearly mutually unbiased to the standard basis. Note also that
$1/z_0=z^*_0$, where $z_0^*$ denotes the complex conjugation of $z_0$.
Now we prove that $\ket{v^{(0)}}$ is unbiased with respect to the Fourier
basis.  For $j = 0$ we have
\begin{eqnarray}
\tilde v_{\tilde 0} & = & \frac{1}{p} \sum_{m = 0}^{p-1} v_m = \frac{1}{p} \left(1+\frac{p-1}{2}(z_0+z_0^*)\right) \nonumber \\
& = & \frac{1}{\sqrt{p}},
\end{eqnarray}
where we have used the fact that the sets $\mathcal{N}$ and
$\mathcal{Q}\setminus\{0\}$ each have $\frac{p-1}{2}$ elements, and
$z_0+z_0^*=2\frac{\sqrt{p}-1}{p-1}$.  To find the remaining
coefficients we will need the following two equations:
\begin{eqnarray}
\sum_{m \in \mathcal{Q}\setminus\{0\}} \omega_p^{- j m } & = & \frac{\sqrt{p}-1}{2}, \\
\sum_{m \in \mathcal{N}} \omega_p^{- j m } & = & \frac{-\sqrt{p}-1}{2},
\end{eqnarray}
which apply for $j \not\equiv 0 \mod p$, and $p\equiv 1\mod4$.
Therefore, for $j\not\equiv 0\mod p$:
\begin{equation}
\tilde v_{\tilde\jmath}= \frac{1}{p} \left(1+ \frac{\sqrt{p}-1}{2}z_0 + \frac{-\sqrt{p}-1}{2}z_0^*\right),
\end{equation}
and hence
\begin{eqnarray}
|\tilde v_{\tilde\jmath} |^2 & = & \frac{1}{p^2} \left|1-\frac{1}{2}(z_0 + z_0^*)+\frac{\sqrt{p}}{2}(z_0 - z_0^*)\right|^2
\nonumber \\
& = & \frac{1}{p^2} \left( 2-2\cos\theta+(p-1)(1 - \cos^2\theta) \right)
\nonumber \\
& = & \frac{1}{p}.
\end{eqnarray}
Thus, the vector $\ket{v^{(0)}}$ is unbiased to the Fourier basis, and
by application of Lemma~\ref{LEM_FOURIER} we have completed the proof.
\end{proof}

Again we have verified by computer-aided algebraic methods that this
set is unextendible in dimension $p = 13$, and conjecture that it is
unextendible in all prime dimensions $p \equiv 1 \mod 4$.  In this way
we have now covered all prime dimensions.

\section{Composite dimensions}

In a composite dimension $d = p_1^{e_1} \dots p_m^{e_m}$ (with prime
factors $p_i\ne p_j$ for $i\ne j$) most of our small sets of
complementary observables follow from the results in
Ref.~\cite{JAlgComb.25.111}.  It is useful to introduce the
Weyl-Heisenberg operators:
\begin{eqnarray}
Z_d \ket{j} & = & \omega_d^j \ket{j}, \\
X_d \ket{j} & = & \ket{(j+1) \bmod d}. 
\end{eqnarray}
It turns out that there always exists a set of
\begin{equation}
\xi = \min_i p_i + 1
\label{A_BOUND}
\end{equation}
complementary observables corresponding to the eigenbases of the
operators
\begin{equation}
Z_d, X_d, X_d Z_d, X_d Z_d^2, \dots X_d Z_d^{\xi-2},
\label{ACW}
\end{equation}
that is ``weakly'' unextendible in the sense that it cannot be
extended by another basis which is an eigenbasis of any other
Weyl-Heisenberg operator.  They are therefore natural candidates for
truly unextendible bases.  Indeed, we verified for all even composite
dimensions presented in Table~\ref{TAB_SUMMARY}, the three eigenbases
of $Z_d$, $X_d$, and $X_dZ_d$ are strongly unextendible.  It is
therefore reasonable to conjecture in general that the complementary
observables \eqref{ACW} form an unextendible set.

In dimensions $9$ and $15$ the bound of Eq.~\eqref{A_BOUND} gives four
bases, but we found unextendible sets with only three measurements.
These sets include the standard basis, its Fourier transform, and a
third basis.  However, as it is difficult to recognise the structure
of the third bases, we will not report them here explicitly.

\section{Continuous variables}
\label{SEC_CV}

One can generalise the notion of mutually unbiased bases to an
infinite dimensional Hilbert space by defining two bases $\{
\ket{\phi} \}$ and $\{ \ket{\psi} \}$ as mutually unbiased if $|
\langle \phi | \psi \rangle |^2 = \kappa$, for all vectors
$\ket{\phi}$ and $\ket{\psi}$, with $\kappa$ a positive constant. Here
we also assume that both bases satisfy the orthonormality condition in
terms of Dirac delta, i.e. $\langle \psi | \psi' \rangle = \delta(\psi
- \psi')$ and $\langle \phi | \phi' \rangle = \delta(\phi - \phi')$.

For an example, consider the Hilbert space $L^2(\mathbb{R})$ and the
generalised eigenstates $\ket{q_\theta}$ of the operator $\hat
q_\theta = \cos \theta \, \hat q + \sin \theta \, \hat p$, where $
\hat q$ and $\hat p$ are the position and momentum operators,
respectively.  The overlap of $\ket{q_\theta}$ with
$\ket{q_{\theta'}}$ can be calculated from the Wigner functions of the
quantum states and is given by~\cite{PhysRevA.70.062101}:
\begin{equation}
|\langle q_\theta | q_{\theta'} \rangle |^2 = \frac{1}{2 \pi \hbar |\sin(\theta - \theta')|}.
\end{equation}
Since the overlap depends only on the angles $\theta$ and $\theta'$,
the bases $\{ \ket{q_{\theta}} \}$ and $\{ \ket{q_{\theta'}} \}$ are
mutually unbiased.  By fixing the first basis to be $\{ \ket{q} \}$,
with $\theta=0$, one can find a symmetric triple of MUBs
\cite{PhysRevA.78.020303}: $\{ \ket{q} \}$, $\{ \ket{q_+} \}$, and $\{
\ket{q_-} \}$, where $\hat q_\pm = \cos(2 \pi/3) \hat q \pm \sin(2 \pi
/ 3) \hat p$.
%, with overlaps given by
%\begin{equation}
%|\langle q | q_\pm \rangle|^2 = |\langle q_- | q_+ \rangle|^2 = \frac{1}{2 \pi \hbar |\sin(2 \pi /3)|}.
%\end{equation}
%A second asymmetric triple of MUBs was also found in \cite{PhysRevA.78.020303} from the eigenstates of $\hat q$, $\hat p$ and $\hat q - \hat p \equiv \sqrt{2} \hat q_{-\pi/4}$.
%Here, the the third operator is multiplied by the factor $\sqrt{2}$ to compensate for the asymmetry of the triple.

It is straightforward to see from the Wigner functions that
multiplying the operator $\hat q_{\theta'}$ by $r$ results in the
overlap $|\langle q_\theta | r q_{\theta'} \rangle|^2 = 1/2\pi \hbar r
|\sin(\theta - \theta')|$.  We now show that by considering
eigenstates of operators of the form $r \hat q_{\theta}$ it is always
possible to extend any pair of MUBs to a triple of MUBs.
\begin{lemma}
If we consider MUBs from the eigenstates of observables $r \hat
q_\theta$, then no unextendible pair of MUBs exists, and all triples
are unextendible.
\end{lemma}
\begin{proof}
We rotate the coordinate system and set the unit of length such that
for the first operator we have $\theta = 0$ and $r=1$.  With this
choice, we now show that MUBs corresponding to the operators $\hat q$
and $r \hat q_\theta = r(\cos\theta \hat q + \sin\theta \hat p)$
always extend to a third MUB for all $\theta \in (0, 2\pi)$, $\theta
\ne \pi$, and $r > 0$.
%We need only consider operators passing through the origin since any alternative pair can be linearly translated to $\hat q$  and $r \hat q_\theta$ while preserving the overlaps between their eigenstates.
%To see how we can construct a set of three MUBs 
We consider a third operator $s \hat q_\phi$ and show that the mutual
unbiasedness equations
\begin{equation}
r |\sin \theta| = s |\sin \phi| = rs |\sin(\theta - \phi)|,
\end{equation}
are satisfied for a particular choice of $s$ and $\phi$.  Choosing $s
= \frac{r |\sin\theta|}{|\sin\phi|}$, the first equality is satisfied
for all $\theta$, $\phi\ne\pi$, and $r$.  The final condition becomes
$|\sin\phi| = r |\sin(\theta - \phi)|$, which is satisfied when
\begin{equation}
\phi_{\pm} = \arctan \left(\frac{\pm r \sin \theta}{1 \pm r \cos \theta} \right).
\label{PHI_PM}
\end{equation}
Hence, the operators $(\hat q, r \hat q_{\theta}, s_\pm \hat
q_{\phi_\pm})$ with $s_{\pm} = \frac{r|\sin\theta|}{\sin(\phi_\pm)}$
and $\phi_\pm$ in Eq.~\eqref{PHI_PM} correspond to two sets of three
MUBs, each with two parameters.  If we consider a fourth basis
corresponding to the operator $t \hat q_\nu$, it is easy to check that
there is no mutually unbiased set of four bases.
%If we restrict ourselves to operators on the unit circle, i.e. $(\hat
%q, \hat q_\theta, \hat q_\phi)$ where $r = s = 1$, then only in the
%special cases $\theta = \pi/3$ and $\theta = 2 \pi /3$ does the pair
%extend to a (symmetric) triple. All other choices of pairs $(\hat q,
%\hat q_\theta)$ form unextendible pairs in this respect.
\end{proof}

For example, choosing $r=1$ and $\theta = 2 \pi /3$ we arrive at two
triples, the symmetric mutually unbiased triple $(\hat q, \hat q_{2
  \pi / 3}, \hat q_{4 \pi /3})$, and $(\hat q, \hat q_{2 \pi /3},
\sqrt{3} \hat q_{\pi/6})$.  Choosing $r=1$ and $\theta = \pi/2$ we
find the asymmetric triples $(\hat q, \hat p, \sqrt{2} \hat q_{\pm
  \pi/4})$~\cite{PhysRevA.78.020303}.

\section{Approaches to unextendibility}

Here we discuss two approaches which reveal that unextendible MUBs are
present in infinitely many dimensions.  Related result can be found in
Refs.~\cite{1502.05245,1604.04797}, but our analysis extends to other
dimensions.

\subsection{Real Hadamards}

Both of our approaches utilise a result which limits the number of
\emph{real} bases in a complete set of MUBs ~\cite{1201.0631}. In
particular, given a complete set of MUBs in dimension $d$, written in
the form $\{ \openone, H_1, \dots, H_d\}$, where $\openone$ is the
identity matrix and $H_i$ are complex Hadamard matrices of order $d$,
then at most one of the Hadamard matrices $H_i$ can be
real. Therefore, any set of complementary observables represented by
the identity and at least two real Hadamard matrices is part of an
unextendible set.

In dimensions $d = 4^i$, constructions given in
Refs.~\cite{IndagMath.35.1,ProcLonMathSoc.75.436,0502024} find $d/2 +
1$ real MUBs (in some of these dimensions Ref.~\cite{1410.4059} finds
a smaller number of real MUBs but still larger than $3$).  These bases
can be transformed into a set of $d/2$ real Hadamard matrices, which
therefore belong to an unextendible set of MUBs. While no additional
real basis is mutually unbiased to this set, $d/2+1$ gives only a
lower bound on the cardinality of the unextendible set. Other complex
Hadamard matrices may exist which represent additional MUBs with
respect to the whole set.

\subsection{Mutually orthogonal Latin squares}

A better lower bound is achieved, and for a larger set of dimensions,
by starting with mutually orthogonal Latin squares (MOLS).
Ref.~\cite{0407081} describes a procedure using Hadamard matrices and
MOLS, which transforms a set of $m$ squares of order $d$ to a set of
$m+2$ MUBs in dimension $d^2$.  Bachelor Latin squares are squares
which have no orthogonal mate (i.e., an unextendible set of MOLS with
$m=1$), and have recently been shown to exist for all $d >
3$~\cite{DCC.40.121,DCC.40.131}.  We therefore use the construction
of~\cite{0407081} with a bachelor Latin square and a real Hadamard
matrix of order $d$. The Hadamard conjecture suggests that such a
matrix always exists for all $d=4s$ and according to most recent
checks, the first case in which a real Hadamard matrix has not yet
been found is $s = 668$~\cite{kharaghani2005hadamard}. With these
ingredients the procedure yields three real MUBs $\{ B_1, B_2, B_3
\}$. To transform this set to the standard form, we multiply each
matrix with the transposed matrix $B_1^T$, i.e. $\{ \openone, B_2
B_1^T, B_3 B_1^T\}$. The matrices $B_2 B_1^T$ and $B_3 B_1^T$ are now
two real Hadamard matrices, and the triple of MUBs does not extend to a
complete set. Unfortunately, this lower bound is not tight. We have
verified in dimension $d = 16$ that the three MUBs obtained in this
way can be extended to at least five bases.

Finally, we utilise the concept of unextendibility to comment briefly
on any conjectured strong relationship between MOLS and MUBs.  It has
been observed that certain complete sets of MOLS translate into
complete sets of MUBs in prime and prime-power dimensions via an
algorithm given in Ref.~\cite{PhysRevA.79.012109}.  Similarly, in a
reverse process, one can construct complete sets of MOLS from certain
complete sets of MUBs~\cite{hall2010mutually}.  It might therefore be
tempting to conjecture a strong statement that to every set of MOLS
there corresponds a set of MUBs.  For the algorithm given in
Ref.~\cite{PhysRevA.79.012109} this is disproved in
Ref.~\cite{PhysScr.T140.014031}. Here we note that such a conjecture
cannot be true for any algorithm directly linking sets of MOLS and
MUB, and preserving (un)extendibility.  This follows from the
observation that a bachelor Latin square should correspond to a set of
three unextendible MUBs in any dimension $d > 3$. However, this is not
the case for dimension five where no unextendible set of MUBs exists
(cf. Section~\ref{SEC_SMALL}). We finish by noting that the weaker
statement given in Ref.~\cite{JOptB.6.L19}, that a maximal set of MUBs
exists in those dimensions $d$ for which one can find a maximal set of
MOLS of order $d$, might still hold true.

\section{Conclusions}

We presented small and unextendible sets of complementary observables,
usually containing three measurements.  It is of course interesting to
ask if even smaller sets exist with the ultimate limit of two.
Mathematically these would correspond to so-called bachelor Hadamard
matrices, which so far are only known to exist in dimension six.  We
therefore hope our study will motivate the search for bachelor
Hadamards in other dimensions as well as the search for further
applications of unextendible MUBs.

\section{Acknowledgments}
We thank \L{}ukasz Rudnicki for pointing out that unextendible MUBs
can distinguish some mixed states from pure states. This work is
supported by the National Research Foundation and Singapore Ministry
of Education Academic Research Fund Tier 2 project MOE2015-T2-2-034.
D. M. acknowledges support by the Institute for Information \&
Communications Technology Promotion (IITP) grant funded by the Korean
government (MSIP) (No. R0190-16-2028, PSQKD).

\bibliographystyle{apsrev4-1}

%\bibliography{UMUB}

\begin{thebibliography}{39}%
\makeatletter
\providecommand \@ifxundefined [1]{%
 \@ifx{#1\undefined}
}%
\providecommand \@ifnum [1]{%
 \ifnum #1\expandafter \@firstoftwo
 \else \expandafter \@secondoftwo
 \fi
}%
\providecommand \@ifx [1]{%
 \ifx #1\expandafter \@firstoftwo
 \else \expandafter \@secondoftwo
 \fi
}%
\providecommand \natexlab [1]{#1}%
\providecommand \enquote  [1]{``#1''}%
\providecommand \bibnamefont  [1]{#1}%
\providecommand \bibfnamefont [1]{#1}%
\providecommand \citenamefont [1]{#1}%
\providecommand \href@noop [0]{\@secondoftwo}%
\providecommand \href [0]{\begingroup \@sanitize@url \@href}%
\providecommand \@href[1]{\@@startlink{#1}\@@href}%
\providecommand \@@href[1]{\endgroup#1\@@endlink}%
\providecommand \@sanitize@url [0]{\catcode `\\12\catcode `\$12\catcode
  `\&12\catcode `\#12\catcode `\^12\catcode `\_12\catcode `\%12\relax}%
\providecommand \@@startlink[1]{}%
\providecommand \@@endlink[0]{}%
\providecommand \url  [0]{\begingroup\@sanitize@url \@url }%
\providecommand \@url [1]{\endgroup\@href {#1}{\urlprefix }}%
\providecommand \urlprefix  [0]{URL }%
\providecommand \Eprint [0]{\href }%
\providecommand \doibase [0]{http://dx.doi.org/}%
\providecommand \selectlanguage [0]{\@gobble}%
\providecommand \bibinfo  [0]{\@secondoftwo}%
\providecommand \bibfield  [0]{\@secondoftwo}%
\providecommand \translation [1]{[#1]}%
\providecommand \BibitemOpen [0]{}%
\providecommand \bibitemStop [0]{}%
\providecommand \bibitemNoStop [0]{.\EOS\space}%
\providecommand \EOS [0]{\spacefactor3000\relax}%
\providecommand \BibitemShut  [1]{\csname bibitem#1\endcsname}%
\let\auto@bib@innerbib\@empty
%</preamble>
\bibitem [{\citenamefont {Durt}\ \emph {et~al.}(2010)\citenamefont {Durt},
  \citenamefont {Englert}, \citenamefont {Bengtsson},\ and\ \citenamefont
  {\.Zyczkowski}}]{IJQI.12.535}%
  \BibitemOpen
  \bibfield  {author} {\bibinfo {author} {\bibfnamefont {T.}~\bibnamefont
  {Durt}}, \bibinfo {author} {\bibfnamefont {B.-G.}\ \bibnamefont {Englert}},
  \bibinfo {author} {\bibfnamefont {I.}~\bibnamefont {Bengtsson}}, \ and\
  \bibinfo {author} {\bibfnamefont {K.}~\bibnamefont {\.Zyczkowski}},\
  }\href@noop {} {\bibfield  {journal} {\bibinfo  {journal} {Int. J. Quant.
  Inf.}\ }\textbf {\bibinfo {volume} {12}},\ \bibinfo {pages} {535} (\bibinfo
  {year} {2010})}\BibitemShut {NoStop}%
\bibitem [{\citenamefont {Grassl}(2004)}]{0406175}%
  \BibitemOpen
  \bibfield  {author} {\bibinfo {author} {\bibfnamefont {M.}~\bibnamefont
  {Grassl}},\ }\href@noop {} {\bibfield  {journal} {\bibinfo  {journal}
  {quant-ph/0406175}\ } (\bibinfo {year} {2004})}\BibitemShut {NoStop}%
\bibitem [{\citenamefont {Boykin}\ \emph {et~al.}(2005)\citenamefont {Boykin},
  \citenamefont {Sitharam}, \citenamefont {Tarifi},\ and\ \citenamefont
  {Wocjan}}]{0502024}%
  \BibitemOpen
  \bibfield  {author} {\bibinfo {author} {\bibfnamefont {P.}~\bibnamefont
  {Boykin}}, \bibinfo {author} {\bibfnamefont {M.}~\bibnamefont {Sitharam}},
  \bibinfo {author} {\bibfnamefont {M.}~\bibnamefont {Tarifi}}, \ and\ \bibinfo
  {author} {\bibfnamefont {P.}~\bibnamefont {Wocjan}},\ }\href@noop {}
  {\bibfield  {journal} {\bibinfo  {journal} {quant-ph/0502024}\ } (\bibinfo
  {year} {2005})}\BibitemShut {NoStop}%
\bibitem [{\citenamefont {Paterek}\ \emph {et~al.}(2009)\citenamefont
  {Paterek}, \citenamefont {Daki\'c},\ and\ \citenamefont
  {Brukner}}]{PhysRevA.79.012109}%
  \BibitemOpen
  \bibfield  {author} {\bibinfo {author} {\bibfnamefont {T.}~\bibnamefont
  {Paterek}}, \bibinfo {author} {\bibfnamefont {B.}~\bibnamefont {Daki\'c}}, \
  and\ \bibinfo {author} {\bibfnamefont {{\v C}.}~\bibnamefont {Brukner}},\
  }\href@noop {} {\bibfield  {journal} {\bibinfo  {journal} {Phys. Rev. A}\
  }\textbf {\bibinfo {volume} {79}},\ \bibinfo {pages} {012109} (\bibinfo
  {year} {2009})}\BibitemShut {NoStop}%
\bibitem [{\citenamefont {Brierley}\ and\ \citenamefont
  {Weigert}(2009)}]{PhysRevA.79.052316}%
  \BibitemOpen
  \bibfield  {author} {\bibinfo {author} {\bibfnamefont {S.}~\bibnamefont
  {Brierley}}\ and\ \bibinfo {author} {\bibfnamefont {S.}~\bibnamefont
  {Weigert}},\ }\href@noop {} {\bibfield  {journal} {\bibinfo  {journal} {Phys.
  Rev. A}\ }\textbf {\bibinfo {volume} {79}},\ \bibinfo {pages} {052316}
  (\bibinfo {year} {2009})}\BibitemShut {NoStop}%
\bibitem [{\citenamefont {Jaming}\ \emph {et~al.}(2009)\citenamefont {Jaming},
  \citenamefont {Matolcsi}, \citenamefont {M\'ora}, \citenamefont
  {Sz\"oll\H{o}si},\ and\ \citenamefont {Weiner}}]{JPA.42.245305}%
  \BibitemOpen
  \bibfield  {author} {\bibinfo {author} {\bibfnamefont {P.}~\bibnamefont
  {Jaming}}, \bibinfo {author} {\bibfnamefont {M.}~\bibnamefont {Matolcsi}},
  \bibinfo {author} {\bibfnamefont {P.}~\bibnamefont {M\'ora}}, \bibinfo
  {author} {\bibfnamefont {F.}~\bibnamefont {Sz\"oll\H{o}si}}, \ and\ \bibinfo
  {author} {\bibfnamefont {M.}~\bibnamefont {Weiner}},\ }\href@noop {}
  {\bibfield  {journal} {\bibinfo  {journal} {J. Phys. A}\ }\textbf {\bibinfo
  {volume} {42}},\ \bibinfo {pages} {245305} (\bibinfo {year}
  {2009})}\BibitemShut {NoStop}%
\bibitem [{\citenamefont {McNulty}\ and\ \citenamefont
  {Weigert}(2012)}]{JPA.45.102001}%
  \BibitemOpen
  \bibfield  {author} {\bibinfo {author} {\bibfnamefont {D.}~\bibnamefont
  {McNulty}}\ and\ \bibinfo {author} {\bibfnamefont {S.}~\bibnamefont
  {Weigert}},\ }\href@noop {} {\bibfield  {journal} {\bibinfo  {journal} {J.
  Phys. A}\ }\textbf {\bibinfo {volume} {45}},\ \bibinfo {pages} {102001}
  (\bibinfo {year} {2012})}\BibitemShut {NoStop}%
\bibitem [{\citenamefont {Goyeneche}(2013)}]{JPhysA.46.105301}%
  \BibitemOpen
  \bibfield  {author} {\bibinfo {author} {\bibfnamefont {D.}~\bibnamefont
  {Goyeneche}},\ }\href@noop {} {\bibfield  {journal} {\bibinfo  {journal} {J.
  Phys. A}\ }\textbf {\bibinfo {volume} {46}},\ \bibinfo {pages} {105301}
  (\bibinfo {year} {2013})}\BibitemShut {NoStop}%
\bibitem [{\citenamefont {Szanto}(2016)}]{1502.05245}%
  \BibitemOpen
  \bibfield  {author} {\bibinfo {author} {\bibfnamefont {A.}~\bibnamefont
  {Szanto}},\ }\href@noop {} {\bibfield  {journal} {\bibinfo  {journal} {L.
  Alg. Appl.}\ }\textbf {\bibinfo {volume} {496}},\ \bibinfo {pages} {392}
  (\bibinfo {year} {2016})}\BibitemShut {NoStop}%
\bibitem [{\citenamefont {Jedwab}\ and\ \citenamefont
  {Yen}(2016)}]{1604.04797}%
  \BibitemOpen
  \bibfield  {author} {\bibinfo {author} {\bibfnamefont {J.}~\bibnamefont
  {Jedwab}}\ and\ \bibinfo {author} {\bibfnamefont {L.}~\bibnamefont {Yen}},\
  }\href@noop {} {\bibfield  {journal} {\bibinfo  {journal} {arXiv:1604.04797}\
  } (\bibinfo {year} {2016})}\BibitemShut {NoStop}%
\bibitem [{\citenamefont {Mandayam}\ \emph {et~al.}(2014)\citenamefont
  {Mandayam}, \citenamefont {Bandyopadhyay}, \citenamefont {Grassl},\ and\
  \citenamefont {Wootters}}]{QuantInfComp.14.823}%
  \BibitemOpen
  \bibfield  {author} {\bibinfo {author} {\bibfnamefont {P.}~\bibnamefont
  {Mandayam}}, \bibinfo {author} {\bibfnamefont {S.}~\bibnamefont
  {Bandyopadhyay}}, \bibinfo {author} {\bibfnamefont {M.}~\bibnamefont
  {Grassl}}, \ and\ \bibinfo {author} {\bibfnamefont {W.~K.}\ \bibnamefont
  {Wootters}},\ }\href@noop {} {\bibfield  {journal} {\bibinfo  {journal}
  {Quant. Inf. Comp.}\ }\textbf {\bibinfo {volume} {14}},\ \bibinfo {pages}
  {823} (\bibinfo {year} {2014})}\BibitemShut {NoStop}%
\bibitem [{\citenamefont {Thas}(2014)}]{1407.2778}%
  \BibitemOpen
  \bibfield  {author} {\bibinfo {author} {\bibfnamefont {K.}~\bibnamefont
  {Thas}},\ }\href@noop {} {\bibfield  {journal} {\bibinfo  {journal}
  {arXiv:1407.2778}\ } (\bibinfo {year} {2014})}\BibitemShut {NoStop}%
\bibitem [{\citenamefont {Hegde}\ and\ \citenamefont
  {Mandayam}(2015)}]{1508.05892}%
  \BibitemOpen
  \bibfield  {author} {\bibinfo {author} {\bibfnamefont {V.}~\bibnamefont
  {Hegde}}\ and\ \bibinfo {author} {\bibfnamefont {P.}~\bibnamefont
  {Mandayam}},\ }\href@noop {} {\bibfield  {journal} {\bibinfo  {journal}
  {arXiv:1508.05892}\ } (\bibinfo {year} {2015})}\BibitemShut {NoStop}%
\bibitem [{\citenamefont {Schwinger}(1960)}]{PNAS.46.560}%
  \BibitemOpen
  \bibfield  {author} {\bibinfo {author} {\bibfnamefont {J.}~\bibnamefont
  {Schwinger}},\ }\href@noop {} {\bibfield  {journal} {\bibinfo  {journal}
  {Proc. Nat. Acad. Sci. USA}\ }\textbf {\bibinfo {volume} {46}},\ \bibinfo
  {pages} {570} (\bibinfo {year} {1960})}\BibitemShut {NoStop}%
\bibitem [{\citenamefont {Korzekwa}\ \emph {et~al.}(2014)\citenamefont
  {Korzekwa}, \citenamefont {Jennings},\ and\ \citenamefont
  {Rudolph}}]{korzekwa2014operational}%
  \BibitemOpen
  \bibfield  {author} {\bibinfo {author} {\bibfnamefont {K.}~\bibnamefont
  {Korzekwa}}, \bibinfo {author} {\bibfnamefont {D.}~\bibnamefont {Jennings}},
  \ and\ \bibinfo {author} {\bibfnamefont {T.}~\bibnamefont {Rudolph}},\
  }\href@noop {} {\bibfield  {journal} {\bibinfo  {journal} {Phys. Rev. A}\
  }\textbf {\bibinfo {volume} {89}},\ \bibinfo {pages} {052108} (\bibinfo
  {year} {2014})}\BibitemShut {NoStop}%
\bibitem [{\citenamefont {Jak\'{o}bczyk}\ and\ \citenamefont
  {Siennicki}(2001)}]{Jakobczyk_01}%
  \BibitemOpen
  \bibfield  {author} {\bibinfo {author} {\bibfnamefont {L.}~\bibnamefont
  {Jak\'{o}bczyk}}\ and\ \bibinfo {author} {\bibfnamefont {M.}~\bibnamefont
  {Siennicki}},\ }\href@noop {} {\bibfield  {journal} {\bibinfo  {journal}
  {Phys. Lett. A}\ }\textbf {\bibinfo {volume} {286}},\ \bibinfo {pages} {383}
  (\bibinfo {year} {2001})}\BibitemShut {NoStop}%
\bibitem [{\citenamefont {Kimura}(2003)}]{Kimura_03}%
  \BibitemOpen
  \bibfield  {author} {\bibinfo {author} {\bibfnamefont {G.}~\bibnamefont
  {Kimura}},\ }\href@noop {} {\bibfield  {journal} {\bibinfo  {journal} {Phys.
  Lett. A}\ }\textbf {\bibinfo {volume} {314}},\ \bibinfo {pages} {339}
  (\bibinfo {year} {2003})}\BibitemShut {NoStop}%
\bibitem [{\citenamefont {Appleby}(2009)}]{Appleby_09}%
  \BibitemOpen
  \bibfield  {author} {\bibinfo {author} {\bibfnamefont {D.~M.}\ \bibnamefont
  {Appleby}},\ }\href@noop {} {\bibfield  {journal} {\bibinfo  {journal} {AIP
  Conf. Proc.}\ }\textbf {\bibinfo {volume} {1101}},\ \bibinfo {pages} {223}
  (\bibinfo {year} {2009})}\BibitemShut {NoStop}%
\bibitem [{\citenamefont {Bennett}\ \emph {et~al.}(1999)\citenamefont
  {Bennett}, \citenamefont {DiVincenzo}, \citenamefont {Mor}, \citenamefont
  {Shor}, \citenamefont {Smolin},\ and\ \citenamefont
  {Terhal}}]{PhysRevLett.82.5385}%
  \BibitemOpen
  \bibfield  {author} {\bibinfo {author} {\bibfnamefont {C.~H.}\ \bibnamefont
  {Bennett}}, \bibinfo {author} {\bibfnamefont {D.~P.}\ \bibnamefont
  {DiVincenzo}}, \bibinfo {author} {\bibfnamefont {T.}~\bibnamefont {Mor}},
  \bibinfo {author} {\bibfnamefont {P.~W.}\ \bibnamefont {Shor}}, \bibinfo
  {author} {\bibfnamefont {J.~A.}\ \bibnamefont {Smolin}}, \ and\ \bibinfo
  {author} {\bibfnamefont {B.~M.}\ \bibnamefont {Terhal}},\ }\href@noop {}
  {\bibfield  {journal} {\bibinfo  {journal} {Phys. Rev. Lett.}\ }\textbf
  {\bibinfo {volume} {82}},\ \bibinfo {pages} {5385} (\bibinfo {year}
  {1999})}\BibitemShut {NoStop}%
\bibitem [{\citenamefont {Bravyi}\ and\ \citenamefont
  {Smolin}(2011)}]{PhysRevA.84.042306}%
  \BibitemOpen
  \bibfield  {author} {\bibinfo {author} {\bibfnamefont {S.}~\bibnamefont
  {Bravyi}}\ and\ \bibinfo {author} {\bibfnamefont {J.~A.}\ \bibnamefont
  {Smolin}},\ }\href@noop {} {\bibfield  {journal} {\bibinfo  {journal} {Phys.
  Rev. A}\ }\textbf {\bibinfo {volume} {84}},\ \bibinfo {pages} {042306}
  (\bibinfo {year} {2011})}\BibitemShut {NoStop}%
\bibitem [{\citenamefont {Brierley}\ \emph {et~al.}(2010)\citenamefont
  {Brierley}, \citenamefont {Weigert},\ and\ \citenamefont
  {Bengtsson}}]{0907.4097}%
  \BibitemOpen
  \bibfield  {author} {\bibinfo {author} {\bibfnamefont {S.}~\bibnamefont
  {Brierley}}, \bibinfo {author} {\bibfnamefont {S.}~\bibnamefont {Weigert}}, \
  and\ \bibinfo {author} {\bibfnamefont {I.}~\bibnamefont {Bengtsson}},\
  }\href@noop {} {\bibfield  {journal} {\bibinfo  {journal} {Quant. Inf.
  Comp.}\ }\textbf {\bibinfo {volume} {10}},\ \bibinfo {pages} {803} (\bibinfo
  {year} {2010})}\BibitemShut {NoStop}%
\bibitem [{\citenamefont {Weiner}(2013)}]{PAMS.141.1963}%
  \BibitemOpen
  \bibfield  {author} {\bibinfo {author} {\bibfnamefont {M.}~\bibnamefont
  {Weiner}},\ }\href@noop {} {\bibfield  {journal} {\bibinfo  {journal} {Proc.
  Amer. Math. Soc.}\ }\textbf {\bibinfo {volume} {141}},\ \bibinfo {pages}
  {1963} (\bibinfo {year} {2013})}\BibitemShut {NoStop}%
\bibitem [{\citenamefont {Haagerup}(1997)}]{Haagerup}%
  \BibitemOpen
  \bibfield  {author} {\bibinfo {author} {\bibfnamefont {U.}~\bibnamefont
  {Haagerup}},\ }\href@noop {} {\emph {\bibinfo {title} {Operator Algebras and
  Quantum Field Theory}}}\ (\bibinfo  {publisher} {Int. Press, Cambridge},\
  \bibinfo {year} {1997})\BibitemShut {NoStop}%
\bibitem [{\citenamefont {Butson}(1962)}]{Butson62}%
  \BibitemOpen
  \bibfield  {author} {\bibinfo {author} {\bibfnamefont {A.~T.}\ \bibnamefont
  {Butson}},\ }\href@noop {} {\bibfield  {journal} {\bibinfo  {journal} {Proc.
  Amer. Math. Soc.}\ }\textbf {\bibinfo {volume} {13}},\ \bibinfo {pages} {894}
  (\bibinfo {year} {1962})}\BibitemShut {NoStop}%
\bibitem [{\citenamefont {Gabidulin}\ and\ \citenamefont
  {Shorin}(2002)}]{ProbInfTrans.38.255}%
  \BibitemOpen
  \bibfield  {author} {\bibinfo {author} {\bibfnamefont {E.~M.}\ \bibnamefont
  {Gabidulin}}\ and\ \bibinfo {author} {\bibfnamefont {V.~V.}\ \bibnamefont
  {Shorin}},\ }\href@noop {} {\bibfield  {journal} {\bibinfo  {journal} {Probl.
  Inf. Trans.}\ }\textbf {\bibinfo {volume} {38}},\ \bibinfo {pages} {255}
  (\bibinfo {year} {2002})}\BibitemShut {NoStop}%
\bibitem [{\citenamefont {Aschbacher}\ \emph {et~al.}(2007)\citenamefont
  {Aschbacher}, \citenamefont {Childs},\ and\ \citenamefont
  {Wocjan}}]{JAlgComb.25.111}%
  \BibitemOpen
  \bibfield  {author} {\bibinfo {author} {\bibfnamefont {M.}~\bibnamefont
  {Aschbacher}}, \bibinfo {author} {\bibfnamefont {A.~M.}\ \bibnamefont
  {Childs}}, \ and\ \bibinfo {author} {\bibfnamefont {P.}~\bibnamefont
  {Wocjan}},\ }\href@noop {} {\bibfield  {journal} {\bibinfo  {journal} {J.
  Alg. Comb.}\ }\textbf {\bibinfo {volume} {25}},\ \bibinfo {pages} {111}
  (\bibinfo {year} {2007})}\BibitemShut {NoStop}%
\bibitem [{\citenamefont {Gibbons}\ \emph {et~al.}(2004)\citenamefont
  {Gibbons}, \citenamefont {Hoffman},\ and\ \citenamefont
  {Wootters}}]{PhysRevA.70.062101}%
  \BibitemOpen
  \bibfield  {author} {\bibinfo {author} {\bibfnamefont {K.~S.}\ \bibnamefont
  {Gibbons}}, \bibinfo {author} {\bibfnamefont {M.~J.}\ \bibnamefont
  {Hoffman}}, \ and\ \bibinfo {author} {\bibfnamefont {W.~K.}\ \bibnamefont
  {Wootters}},\ }\href@noop {} {\bibfield  {journal} {\bibinfo  {journal}
  {Phys. Rev. A}\ }\textbf {\bibinfo {volume} {70}},\ \bibinfo {pages} {062101}
  (\bibinfo {year} {2004})}\BibitemShut {NoStop}%
\bibitem [{\citenamefont {Weigert}\ and\ \citenamefont
  {Wilkinson}(2008)}]{PhysRevA.78.020303}%
  \BibitemOpen
  \bibfield  {author} {\bibinfo {author} {\bibfnamefont {S.}~\bibnamefont
  {Weigert}}\ and\ \bibinfo {author} {\bibfnamefont {M.}~\bibnamefont
  {Wilkinson}},\ }\href@noop {} {\bibfield  {journal} {\bibinfo  {journal}
  {Phys. Rev. A}\ }\textbf {\bibinfo {volume} {78}},\ \bibinfo {pages} {020303}
  (\bibinfo {year} {2008})}\BibitemShut {NoStop}%
\bibitem [{\citenamefont {Matolcsi}\ \emph {et~al.}(2013)\citenamefont
  {Matolcsi}, \citenamefont {Ruzsa},\ and\ \citenamefont {Weiner}}]{1201.0631}%
  \BibitemOpen
  \bibfield  {author} {\bibinfo {author} {\bibfnamefont {M.}~\bibnamefont
  {Matolcsi}}, \bibinfo {author} {\bibfnamefont {I.~Z.}\ \bibnamefont {Ruzsa}},
  \ and\ \bibinfo {author} {\bibfnamefont {M.}~\bibnamefont {Weiner}},\
  }\href@noop {} {\bibfield  {journal} {\bibinfo  {journal} {Australas. J.
  Combin.}\ }\textbf {\bibinfo {volume} {55}},\ \bibinfo {pages} {35} (\bibinfo
  {year} {2013})}\BibitemShut {NoStop}%
\bibitem [{\citenamefont {Cameron}\ and\ \citenamefont
  {Seidel}(1973)}]{IndagMath.35.1}%
  \BibitemOpen
  \bibfield  {author} {\bibinfo {author} {\bibfnamefont {P.~J.}\ \bibnamefont
  {Cameron}}\ and\ \bibinfo {author} {\bibfnamefont {J.~J.}\ \bibnamefont
  {Seidel}},\ }\href@noop {} {\bibfield  {journal} {\bibinfo  {journal} {Indag.
  Math.}\ }\textbf {\bibinfo {volume} {35}},\ \bibinfo {pages} {1} (\bibinfo
  {year} {1973})}\BibitemShut {NoStop}%
\bibitem [{\citenamefont {Calderbank}\ \emph {et~al.}(1997)\citenamefont
  {Calderbank}, \citenamefont {Cameron}, \citenamefont {Kantor},\ and\
  \citenamefont {Seidel}}]{ProcLonMathSoc.75.436}%
  \BibitemOpen
  \bibfield  {author} {\bibinfo {author} {\bibfnamefont {A.~R.}\ \bibnamefont
  {Calderbank}}, \bibinfo {author} {\bibfnamefont {P.~J.}\ \bibnamefont
  {Cameron}}, \bibinfo {author} {\bibfnamefont {W.~M.}\ \bibnamefont {Kantor}},
  \ and\ \bibinfo {author} {\bibfnamefont {J.~J.}\ \bibnamefont {Seidel}},\
  }\href@noop {} {\bibfield  {journal} {\bibinfo  {journal} {Proc. London Math.
  Soc.}\ }\textbf {\bibinfo {volume} {75}},\ \bibinfo {pages} {436} (\bibinfo
  {year} {1997})}\BibitemShut {NoStop}%
\bibitem [{\citenamefont {Gow}(2014)}]{1410.4059}%
  \BibitemOpen
  \bibfield  {author} {\bibinfo {author} {\bibfnamefont {R.}~\bibnamefont
  {Gow}},\ }\href@noop {} {\bibfield  {journal} {\bibinfo  {journal}
  {arXiv:1410.4059}\ } (\bibinfo {year} {2014})}\BibitemShut {NoStop}%
\bibitem [{\citenamefont {Wocjan}\ and\ \citenamefont {Beth}(2005)}]{0407081}%
  \BibitemOpen
  \bibfield  {author} {\bibinfo {author} {\bibfnamefont {P.}~\bibnamefont
  {Wocjan}}\ and\ \bibinfo {author} {\bibfnamefont {T.}~\bibnamefont {Beth}},\
  }\href@noop {} {\bibfield  {journal} {\bibinfo  {journal} {Quant. Inf.
  Comp.}\ }\textbf {\bibinfo {volume} {5}},\ \bibinfo {pages} {93} (\bibinfo
  {year} {2005})}\BibitemShut {NoStop}%
\bibitem [{\citenamefont {Evans}(2006)}]{DCC.40.121}%
  \BibitemOpen
  \bibfield  {author} {\bibinfo {author} {\bibfnamefont {A.~B.}\ \bibnamefont
  {Evans}},\ }\href@noop {} {\bibfield  {journal} {\bibinfo  {journal} {Des.
  Codes. Crypt.}\ }\textbf {\bibinfo {volume} {40}},\ \bibinfo {pages} {121}
  (\bibinfo {year} {2006})}\BibitemShut {NoStop}%
\bibitem [{\citenamefont {Wanless}\ and\ \citenamefont
  {Debb}(2006)}]{DCC.40.131}%
  \BibitemOpen
  \bibfield  {author} {\bibinfo {author} {\bibfnamefont {I.~M.}\ \bibnamefont
  {Wanless}}\ and\ \bibinfo {author} {\bibfnamefont {B.~S.}\ \bibnamefont
  {Debb}},\ }\href@noop {} {\bibfield  {journal} {\bibinfo  {journal} {Des.
  Codes. Crypt.}\ }\textbf {\bibinfo {volume} {40}},\ \bibinfo {pages} {131}
  (\bibinfo {year} {2006})}\BibitemShut {NoStop}%
\bibitem [{\citenamefont {Kharaghani}\ and\ \citenamefont
  {Tayfeh-Rezaie}(2005)}]{kharaghani2005hadamard}%
  \BibitemOpen
  \bibfield  {author} {\bibinfo {author} {\bibfnamefont {H.}~\bibnamefont
  {Kharaghani}}\ and\ \bibinfo {author} {\bibfnamefont {B.}~\bibnamefont
  {Tayfeh-Rezaie}},\ }\href@noop {} {\bibfield  {journal} {\bibinfo  {journal}
  {J. Combin. Des.}\ }\textbf {\bibinfo {volume} {13}},\ \bibinfo {pages} {435}
  (\bibinfo {year} {2005})}\BibitemShut {NoStop}%
\bibitem [{\citenamefont {Hall}\ and\ \citenamefont
  {Rao}(2010)}]{hall2010mutually}%
  \BibitemOpen
  \bibfield  {author} {\bibinfo {author} {\bibfnamefont {J.~L.}\ \bibnamefont
  {Hall}}\ and\ \bibinfo {author} {\bibfnamefont {A.}~\bibnamefont {Rao}},\
  }\href@noop {} {\bibfield  {journal} {\bibinfo  {journal} {J. Phys. A}\
  }\textbf {\bibinfo {volume} {43}},\ \bibinfo {pages} {135302} (\bibinfo
  {year} {2010})}\BibitemShut {NoStop}%
\bibitem [{\citenamefont {Paterek}\ \emph {et~al.}(2010)\citenamefont
  {Paterek}, \citenamefont {Paw{\l}owski}, \citenamefont {Grassl},\ and\
  \citenamefont {Brukner}}]{PhysScr.T140.014031}%
  \BibitemOpen
  \bibfield  {author} {\bibinfo {author} {\bibfnamefont {T.}~\bibnamefont
  {Paterek}}, \bibinfo {author} {\bibfnamefont {M.}~\bibnamefont
  {Paw{\l}owski}}, \bibinfo {author} {\bibfnamefont {M.}~\bibnamefont
  {Grassl}}, \ and\ \bibinfo {author} {\bibfnamefont {{\v C}.}~\bibnamefont
  {Brukner}},\ }\href@noop {} {\bibfield  {journal} {\bibinfo  {journal} {Phys.
  Scr.}\ }\textbf {\bibinfo {volume} {T140}},\ \bibinfo {pages} {014031}
  (\bibinfo {year} {2010})}\BibitemShut {NoStop}%
\bibitem [{\citenamefont {Saniga}\ \emph {et~al.}(2004)\citenamefont {Saniga},
  \citenamefont {Planat},\ and\ \citenamefont {Rosu}}]{JOptB.6.L19}%
  \BibitemOpen
  \bibfield  {author} {\bibinfo {author} {\bibfnamefont {M.}~\bibnamefont
  {Saniga}}, \bibinfo {author} {\bibfnamefont {M.}~\bibnamefont {Planat}}, \
  and\ \bibinfo {author} {\bibfnamefont {H.}~\bibnamefont {Rosu}},\ }\href@noop
  {} {\bibfield  {journal} {\bibinfo  {journal} {J. Opt. B}\ }\textbf {\bibinfo
  {volume} {6}},\ \bibinfo {pages} {L19} (\bibinfo {year} {2004})}\BibitemShut
  {NoStop}%
\end{thebibliography}

%merlin.mbs apsrev4-1.bst 2010-07-25 4.21a (PWD, AO, DPC) hacked
%Control: key (0)
%Control: author (72) initials jnrlst
%Control: editor formatted (1) identically to author
%Control: production of article title (-1) disabled
%Control: page (0) single
%Control: year (1) truncated
%Control: production of eprint (0) enabled
%

\end{document}